\documentclass[11pt]{article}
\usepackage{amsmath}
\usepackage{amsfonts}
\usepackage{amsthm}
\usepackage{amsmath,amssymb}
\usepackage{graphicx}
\usepackage{color}
\usepackage{url}
\usepackage{graphicx}
\usepackage{appendix}
\usepackage{booktabs}

\renewcommand{\arraystretch}{1.2}
\allowdisplaybreaks[4]

\newtheorem{theorem}{Theorem}
\newtheorem{lemma}{Lemma}
\newtheorem{corollary}{Corollary}

\newtheorem{definition}{Definition}
\newtheorem{proposition}{Proposition}

\newcommand{\RNum}[1]{\uppercase\expandafter{\romannumeral #1\relax}}

\newcommand{\gf}{{\mathbb{F}}}

\newcommand{\ls}[1]
    {\dimen0=\fontdimen6\the\font\lineskip=#1\dimen0
     \advance\lineskip.5\fontdimen5\the\font
     \advance\lineskip-\dimen0
     \lineskiplimit=0.9\lineskip
     \baselineskip=\lineskip
     \advance\baselineskip\dimen0
     \normallineskip\lineskip\normallineskiplimit\lineskiplimit
     \normalbaselineskip\baselineskip
     \ignorespaces}

\begin{document}

\bibliographystyle{abbrv}

\title{ In-depth analysis of S-boxes over binary finite fields concerning their differential and Feistel boomerang differential uniformities}

\author{Yuying Man\footnotemark[1]\thanks{Y. Man and X. Zeng are with Hubei Key Laboratory of Applied Mathematics, Faculty of Mathematics and Statistics, Hubei University, Wuhan 430062, China. Email:yuying.man@aliyun, xzeng@hubu.edu.cn},
Sihem Mesnager\footnotemark[2]\thanks{S. Mesnager is with Department of Mathematics, University of Paris 8, F-93526 Saint-Denis, Paris, France, the Laboratory of Analysis, Geometry, and Applications (LAGA), University Sorbonne Paris Nord CNRS, UMR 7539, F-93430, Villetaneuse, France, and also with the T\'{e}l\'{e}com Paris, Palaiseau, France. Email: smesnager@univ-paris8.fr},
Nian Li\footnotemark[3]\thanks{N. Li is with Hubei Key Laboratory of Applied Mathematics, School of Cyber Science and Technology, Hubei University, Wuhan 430062, China. Email: nian.li@hubu.edu.cn},
Xiangyong Zeng \footnotemark[1],
Xiaohu Tang\footnotemark[4]\thanks{X. Tang is with the Information Coding \& Transmission Key Lab of Sichuan Province, CSNMT Int. Coop. Res. Centre (MoST), Southwest Jiaotong University, Chengdu, 610031, China. Email: xhutang@swjtu.edu.cn}}
\date{\today}
\maketitle

\thispagestyle{plain} \setcounter{page}{1}

\begin{abstract}

Substitution boxes (S-boxes) play a significant role in ensuring the resistance of block ciphers against various attacks. The Difference Distribution Table (DDT), the Feistel Boomerang Connectivity Table (FBCT), the Feistel Boomerang Difference Table (FBDT) and the Feistel Boomerang Extended Table (FBET) of a given S-box are crucial tools to analyze its security concerning specific attacks. However, the results on them are rare. In this paper, we investigate the properties of the power function $F(x):=x^{2^{m+1}-1}$ over the finite field $\gf_{2^n}$ of order $2^n$ where $n=2m$ or $n=2m+1$  ($m$ stands for a positive integer). As a consequence, by carrying out certain finer manipulations of solving specific equations over $\gf_{2^n}$, we give explicit values of all entries of the DDT, the FBCT, the FBDT and the FBET of the investigated power functions.  From the theoretical point of view, our study pushes further former investigations on differential and Feistel boomerang differential uniformities for a novel power function $F$. From a cryptographic point of view, when considering  Feistel block cipher involving $F$, our in-depth analysis helps select $F$ resistant to differential attacks, Feistel differential attacks and Feistel boomerang attacks, respectively.



\end{abstract}

\section{Introduction}

Vectorial Boolean functions are widely applied to the design of block and stream ciphers in symmetric cryptography. In that context, they are called \emph{S-boxes}
(substitution boxes). They are essential components of symmetric cryptographic algorithms since S-boxes usually are the only nonlinear elements of these cryptosystems. Differential uniformity measures the resistance of S-boxes against differential attacks. Specifically, an S-box in a block cipher is a vectorial function that takes $n$ binary inputs and whose image is a binary $m$-tuple for some positive integers $n$ and $m$. Generally, we consider these S-boxes over binary finite fields after identifying the vector space $\mathbb{F}_{2}^n$ over  $\mathbb{F}_{2}$ with the finite field  $\mathbb{F}_{2^n}$ of order $2^n$. This allows us to use the rich theory of finite fields and its machinery. S-boxes are usually the only nonlinear elements of the most modern block ciphers. It is, therefore, essential to employ S-boxes with good cryptographic properties to resist various attacks.
An excellent book that covers this topic is \cite{CC-book-2021}. Differential attack, introduced by Biham and Shamir \cite{diff-at} in 1991, is one of the most fundamental cryptanalytic tools to assess the security of block ciphers. The Difference Distribution Table (DDT) and the differential uniformity of S-boxes, introduced by Nyberg \cite{def-DDT} in 1993, can be used to measure the ability of an S-box to resist the differential attack. The smaller the differential uniformity $\delta_F$ of an S-box $F$, the stronger its ability to resist differential attack. Typically, the optimal functions satisfy $\delta_F=2$ and are called almost perfect nonlinear (APN). However, for an even $n$, functions defined over $\mathbb{F}_{2^n}$ with a small differential uniformity are scarce, and only one example of an APN permutation is known for $n=6$~\cite{BROWNING}. In particular, when $n$ is even, a few classes of infinite families of functions have differential uniformity $\delta_F=4$. Some of such functions are power functions of the form $F(x)=x^d$ for some $d$. The DDT of an S-box is essential in estimating a differential characteristic's probability. For an input difference $a\in \gf_{2^n}$ and an output difference $b\in \gf_{2^n}$, it counts the number of cases when the input difference of a pair is $a$, and the output difference is $b$ \cite{diff-at}. The differential spectrum describes the frequencies of the values appearing in the DDT of an S-box. It can provide helpful information for finding the differential path in the differential analysis. As a particular class of functions over finite fields, power functions, namely monomial functions, have been extensively studied in the last decades due to their simple algebraic form and lower implementation cost in a hardware environment. Thus, determining the differential spectra of power functions with low differential uniformity has attracted much attention. However, determining the differential spectra of power functions is relatively challenging. With known differential spectra, only a few power functions over $\gf_{2^n}$ exist. For articles dealing with recent results in this context, the readers can refer to the following list and references therein (\cite{diff-IJIC10}, \cite{diff-IT11}, \cite{diff-DCC14}, \cite{Kim-Mesnager-diff-FFA22},\cite{Kim-Mesnager-diff-FFA23}, \cite{diff-IT23}, \cite{diff-FFA17}, \cite{diff-DCC18}). It is worth noting that determining the DDT of a power function is much more complex than finding its differential spectrum. To our best knowledge, there are only three classes of power functions over $\gf_{2^n}$ with known DDT (see Table \ref{table-1}).

\begin{table}
\caption{Power mappings $F(x)=x^d$ over $\gf_{2^n}$ with known DDT}\label{table-1}
\renewcommand\arraystretch{1.2}
\setlength\tabcolsep{10pt}
\centering
\begin{tabular}{cccc}
		\toprule
      $d$ & Condition & differential uniformity &  Ref. \\
       \midrule
        $2^n-2$ & $n$ odd or $n$ even & 2 or 4  &  \cite{ex-FBCT} \\
        $2^k+1$ & ${\rm gcd}(n,k)=d$ & $2^d$ & \cite{ex-FBCT}\\
       $2^{3k}+2^{2k}+2^k-1$ & $n=4k$ & $2^{2k}$ & \cite{ex-FBCT, diff-IWSDA} \\
       $2^{m+1}-1$ & $n=2m+1$ or $n=2m$ & 2 or $2^m$ & This paper \\
\bottomrule
\end{tabular}
\end{table}

Boomerang attack is another crucial cryptanalytical technique on block cyphers, introduced by Wagner \cite{BA} in 1999, a differential cryptanalysis variant. To analyze the boomerang attack of block cyphers in a better way, analogous to the DDT concerning the differential attack, in Eurocrypt 2018, Cid, Huang, Peyrin, Sasaki and Song in \cite{BCT} introduced a new tool known as Boomerang Connectivity Table (BCT) to measure the resistance of an S-box against boomerang attacks. The BCT tool allows us to easily evaluate the probability of the boomerang switches when it covers one round. Small entries in the BCT of a cipher prevent it from attacks related to the boomerang cryptanalysis. It is well-known that a function's set of possible BCTs is preserved under the inversion and the affine equivalence~\cite{BCT-u}. Besides, a function $F$ with optimal behaviours concerning differential attacks is also optimal concerning boomerang attacks (see \cite{BCT}). When studying how to extend the BCT theory to boomerang switches on more rounds, Wang and Peyrin \cite{BDT} proposed the concept of the Boomerang Difference Table (BDT), a variant of the BCT with one supplementary variable fixed. However, these works only addressed the case of the Substitution-Permutation Network (SPN) structure and completely neglected the ciphers following a Feistel Network structure. At the same time, it cannot be denied that it is an equally important type of block cipher design, instantiated by the widely used 3-DES and by CLEFIA \cite{Feistel}. To address this deficiency, Boukerrou, Huynh, Lallemand, Mandal, and Minier \cite{def-FBCT} introduced the notion of Feistel Boomerang Connectivity Table (FBCT, the Feistel counterpart of the BCT) that covers switches of one round and then introduced the Feistel Boomerang Difference Table (FBDT, the Feistel counterpart of the BDT) to deal with a 2-round switch and finally proposed a new tool called Feistel Boomerang Extended Table (FBET) that treats the case of an arbitrary number of rounds. FBCT, FBDT and FBET are crucial tables for the analysis of the resistance of block ciphers to powerful attacks such as differential and boomerang attacks. The computation of these explicit values of all entries and the cardinalities of their corresponding sets of such values in each table aimed to facilitate the analysis of differential and boomerang cryptanalysis of S-boxes when studying distinguishers and trails \cite{sur}. The explicit values of the entries of these tables and their cardinalities are crucial tools to test the resistance of block ciphers based on variants of the function against cryptanalytics such as differential and boomerang attacks. However, obtaining these entries and the cardinalities in each table for a given S-box is challenging. To our best knowledge, there are only three classes of functions with known FBCT (see Table \ref{table-2}) and two classes of functions with known FBDT (see Table \ref{table-3}).

\begin{table}
\caption{Power mappings $F(x)=x^d$ over $\gf_{2^n}$ with known FBCT}\label{table-2}
\renewcommand\arraystretch{1.2}
\setlength\tabcolsep{10pt}
\centering
\begin{tabular}{cccc}
		\toprule
      $d$ & Condition & Feistel boomerang uniformity &  Ref. \\
       \midrule
        $2^n-2$ & $n$ odd or $n$ even & 0 or 4  &  \cite{ex-FBCT} \\
        $2^k+1$ & ${\rm gcd}(n,k)=d$ & 0 & \cite{ex-FBCT}\\
       $2^{2k}+2^k+1$ & $n=4k$ & $2^{2k}$ & \cite{ex-FBCT} \\
       $2^{m+1}-1$ & $n=2m+1$ or $n=2m$ & 0 or $2^m$ & This paper \\
\bottomrule
\end{tabular}
\end{table}

\begin{table}
\caption{Power mappings $F(x)=x^d$ over $\gf_{2^n}$ with known FBDT}\label{table-3}
\renewcommand\arraystretch{1.2}
\setlength\tabcolsep{4.5pt}
\centering
\begin{tabular}{cccc}
		\toprule
      $d$ & Condition & Feistel boomerang differential uniformity &  Ref. \\
       \midrule
        $2^n-2$ & $n$ odd or $n$ even & 2 or 4  &  \cite{ex-FBCT} \\
        $2^k+1$ & ${\rm gcd}(n,k)=d$ & $2^d$ & \cite{ex-FBCT}\\
       $2^{m+1}-1$ & $n=2m+1$ or $n=2m$ & 2 or $2^m$ & This paper \\
\bottomrule
\end{tabular}
\end{table}

In this paper, we investigate the specific tables of the power mapping $F(x)=x^{2^{m+1}-1}$  by carrying out some finer manipulations of solving certain equations over $\mathbb{F}_{2^n}$, where $n=2m$ or $n=2m+1$.
Firstly, utilizing the solvability of the quadratic equations and the balance property of the trace function, we compute the values of all entries in the DDT of $F(x)$ and then characterize the differential properties of $F(x)$ by using its DDT.
Secondly, using the technique of reducing the high-degree equations to the low-degree equations and the iterative method, we determine all entries of the FBCT of $F(x)$ entirely and give an accurate value of the number of elements $(a,b)\in \gf_{2^n}^2$ in the spectrum of the FBCT of $F(x)$ (see Definition \ref{definition-FBCT}).
Based on the DDT and FBCT of $F(x)$, we then compute all values of the entries in its FBDT and count the number of elements $(a,c,b)\in \gf_{2^n}^3$ in a given spectrum (see Definition \ref{definition-FBDT}).
Finally, based on the techniques for solving equations over finite fields and the related results of the DDT, FBCT and FBDT of $F(x)$, we provide all explicit values of all entries of the FBET and calculate the number of elements $(a,c,b,d)\in \gf_{2^n}^4$ taking any value in the FBET (see Definition \ref{definition-FBET}).
It is worth mentioning that the power mapping $F(x)=x^{2^{m+1}-1}$ over $\mathbb{F}_{2^n}$ in this paper is the fourth, the third, and the first class of functions with known FBCT, FBDT and FBET, respectively (see Table \ref{table-2}, Table \ref{table-3} and Table \ref{table-4}).

\begin{table}
\caption{Power mapping $F(x)=x^d$ over $\gf_{2^n}$ with known FBET}\label{table-4}
\renewcommand\arraystretch{1.2}
\setlength\tabcolsep{4.5pt}
\centering
\begin{tabular}{cccc}
		\toprule
      $d$ & Condition & Feistel boomerang extended uniformity &  Ref. \\
       \midrule
       $2^{m+1}-1$ & $n=2m+1$ or $n=2m$ & 2 or $2^m$ & This paper \\
\bottomrule
\end{tabular}
\end{table}

The remainder of this paper is organized as follows. In Section \ref{pre}, we first fix our notation and present some
basic notions and a few known results helpful in the technical part of the paper. Next, Section \ref{DDT-result1} provides explicit values of all entries in the DDT of $F(x)$. Section \ref{FBCT-result1}, Section \ref{FBDT-result1} and Section \ref{FBET-result1} offer a detailed study of the FBCT, FBDT and FBET of $F(x)$, respectively. Section \ref{con-remarks} concludes this paper.

\section{Preliminaries}\label{pre}

Throughout this paper, $\mathbb{F}_{2^n}$ denotes the finite field of order $2^n$ and   ${\rm Tr}_{1}^n(\cdot)$ denotes the  (absolute) trace function from $\mathbb{F}_{2^n}$ onto its prime field $\mathbb{F}_{2}$ (where $n$ is a positive integer). Recall that  for $x\in\mathbb{F}_{2^n}$,
${\rm Tr}_{1}^n(x)=\sum_{i=0}^{n-1}x^{2^i}$. For any (finite) set $E$, $|E|$ denotes its cardinality.

In this section,  we recall some basic definitions and present some results which will be used frequently in this paper.

 \begin{definition}\label{definition-DDT}\rm (\cite{def-DDT})
 Let $F(x)$ be a  mapping from $\mathbb{F}_{2^n}$ to itself. The Difference Distribution Table (DDT) of $F(x)$ is a $2^n \times 2^n$ table where the entry at $(a,b)\in \mathbb{F}_{2^n}^2$ is defined by
$${\rm DDT}_{F}(a,b)=|\{x \in \gf_{2^n}: F(x)+F(x+a)=b \}|.$$

The mapping $F(x)$ is said to be \emph{differentially $\delta$-uniform} if $\delta(F)=\delta$ \cite{def-DDT}, and accordingly $\delta(F)$ is called the \emph{differential uniformity} of $F(x)$, where
$$
\delta(F)=\max _{a,b \in \mathbb{F}_{2^n},a\ne 0} {\rm DDT}_{F}(a,b).
$$
When $F$ is used as an S-box inside a cryptosystem, the smaller the value $\delta(F)$ is, the better the contribution of $F$ to the resistance against differential attack.

Assume that a mapping $F(x)$ over $\gf_{2^n}$ has differential uniformity $\delta$  and denote
\[\omega_i=|\left\{a\in \gf_{2^n}^*, b\in \gf_{2^n}\mid {\rm DDT}_{F}(a,b)=i\right\}|,\,\,0\leq i\leq \delta.\]
The differential spectrum of $F$ is defined to be
an ordered sequence
\[
%\mathbb{S}=\left\{(i: \omega_i)\,| \, 0\leq i \leq \delta\right\}.
\mathbb{S} = [\omega_0, \omega_1, \ldots, \omega_{\delta}].
\]

According to the definition of DDT, we have the following identities
\begin{equation}\label{prop}
\sum\limits_{i=0}^\delta \omega_i=2^n(2^n-1)\,\,{\rm and}\,\,\sum\limits_{i=0}^\delta \left(i\times \omega_i\right)=2^n(2^n-1).
\end{equation}
\end{definition}

Let $F: \mathbb{F}_2^n\rightarrow \mathbb{F}_2^n$ be a permutation over $\mathbb{F}_{2^n}$.
The boomerang Connectivity Table (BCT) helps measure the resistance of an S-box $F$  against boomerang attacks.  It is defined more specifically as follows by the  $2^n\times 2^n$ table defined for $a,b\in \mathbb{F}_{2^n}$ by
$$
{\rm BCT}_F(a,b)=|\{x\in \mathbb{F}_{2^n} : F^{-1}(F(x)+b)+F^{-1}(F(x+a)+b)=a\}|.
$$

The boomerang uniformity~(see. \cite{BCT-u}) of $F$ is defined by
$$
\beta_F=\max_{a,b\in \mathbb{F}_{2^n},ab\neq 0}{\rm BCT}_F(a,b).
$$

To consider the case of ciphers following a Feistel Network structure, some variants of the DDT and the BCT for ciphers were presented by Boukerrou et al. \cite{def-FBCT}. They defined the FBCT and the Feistel
boomerang uniformity of $F(x)$ as follows.

\begin{definition}\label{definition-FBCT}\rm (\cite{def-FBCT})
Let $F(x)$ be a mapping from $\mathbb{F}_{2^n}$ to itself. The Feistel Boomerang Connectivity Table (FBCT) is a $2^n\times 2^n$ table defined for $(a,b)\in \mathbb{F}_{2^n}^2$ by
$${\rm FBCT}_{F}(a,b)=|\{x \in \gf_{2^n}: F(x)+F(x+a)+F(x+b)+F(x+a+b)=0 \}|.$$
Clearly, the FBCT satisfies ${\rm FBCT}_{F}(a,b)=2^n$ if $ab(a+b)=0$. Hence, the Feistel boomerang uniformity of $F(x)$ is defined by
$$
\beta_c(F)=\max _{a, b \in \mathbb{F}_{2^n}, ab(a+b)\ne 0} {\rm FBCT}_{F}(a,b).
$$
\end{definition}

Boukerrou et al. \cite{def-FBCT} verified that the counterpart of BDT for the Feistel case is useful for studying a switch over two rounds. Following the idea of BDT, they defined the Feistel Boomerang Difference Table.

The basic properties of the FBCT are studied in~\cite{def-FBCT}. Typically, the FBCT satisfies the following properties.
\begin{itemize}
\item Symmetry: ${\rm FBCT}_F(a,b)={\rm FBCT}_F(b,a)$ for all $a,b\in \mathbb{F}_{2^n}$.
\item Multiplicity: ${\rm FBCT}_F(a,b)\equiv 0\pmod 4$ for all $a,b\in \mathbb{F}_{2^n}$.
\item First line: ${\rm FBCT}_F(0,b)=2^n$ for all $b\in \mathbb{F}_{2^n}$.
\item First column: ${\rm FBCT}_F(a,0)=2^n$ for all $a\in \mathbb{F}_{2^n}$.
\item Diagonal: ${\rm FBCT}_F(a,a)=2^n$ for all $a\in \mathbb{F}_{2^n}$.
\item Equalities: ${\rm FBCT}_F(a,a)={\rm FBCT}_F(a,a+b)$ for all $a,b\in \mathbb{F}_{2^n}$.
\end{itemize}
Another important property of the FBCT is that $F$ is an APN (almost perfect nonlinear) function if and only if ${\rm FBCT}_F(a,b)=0$ for $a,b\in \mathbb{F}_{2^n}$ with $ab(a+b)\neq 0$.
Also, it is important to notice that, generally, $F$ and $F^{-1}$ do not have the same FBCT.

\begin{definition}\label{definition-FBDT}\rm (\cite{def-FBCT})
Let $F(x)$ be a mapping from $\mathbb{F}_{2^n}$ to itself. The Feistel Boomerang Difference Table (FBDT) is a $2^n\times 2^n\times 2^n$ table defined for $(a,c,b)\in \mathbb{F}_{2^n}^3$ by
\begin{equation*}
\begin{aligned}
{\rm FBDT}_{F}(a,c,b)=|\{x \in \gf_{2^n}: &F(x)+F(x+a)+F(x+b)+F(x+a+b)=0,\\
&F(x)+F(x+a)=c \}|.
\end{aligned}
\end{equation*}
It can be seen that ${\rm FBDT}_{F}(a,c,b)=2^n$ if and only if $a=c=0$. Thus, Eddahmani  and  Mesnager (\cite{ex-FBCT}) introduced the Feistel boomerang differential uniformity of $F(x)$ as
$$
\beta_d(F)=\max _{(a, c, b) \in \mathbb{F}_{2^n}^3, (a,c)\ne (0,0)} {\rm FBDT}_{F}(a,c, b).
$$
\end{definition}

Boukerrou et al. \cite{def-FBCT} also proposed the FBET that treats the case of an arbitrary number of rounds. They defined the Feistel Boomerang Extended Table as below.

\begin{definition}\label{definition-FBET}\rm (\cite{def-FBCT})
Let $F(x)$ be a mapping from $\mathbb{F}_{2^n}$ to itself. The Feistel Boomerang Extended Table (FBET) is a $2^n\times 2^n\times 2^n\times 2^n$ table defined for $(a,c,b,d)\in \mathbb{F}_{2^n}^4$ by
\begin{equation*}
\begin{aligned}
{\rm FBET}_{F}(a,c,b,d)=|\{x \in \gf_{2^n}: &F(x)+F(x+a)+F(x+b)+F(x+a+b)=0,\\
&F(x)+F(x+a)=c, F(x+a)+F(x+a+b)=d \}|.
\end{aligned}
\end{equation*}
 It can be seen that ${\rm FBET}_{F}(a,c,b,d)=2^n$ only if $a=c=b=d=0$. This motivates us to introduce the Feistel boomerang extended uniformity $\beta_e(F)$ of $F(x)$ as
$$
\beta_e(F)=\max _{(a, c, b, d) \in \mathbb{F}_{2^n}^4\setminus \{(0,0,0,0)\}} {\rm FBET}_{F}(a,c,b,d).
$$
\end{definition}

The following lemma will be used frequently in this paper.

\begin{lemma}\label{lemma1-root}\rm (\cite{qua-root})
Let $a, b, c \in \mathbb{F}_{2^n}$, $a\ne 0$ and $F(x)=ax^2+bx+c$. Then
\begin{itemize}
\item [\rm (i)] $F(x)$ has exactly one root in $\mathbb{F}_{2^n}$ if and only if $b=0$;
\item [\rm (ii)] $F(x)$ has exactly two roots in $\mathbb{F}_{2^n}$ if and only if $b\ne 0$ and ${\rm Tr}_{1}^n(\frac{ac}{b^2})=0$,
\item [\rm (iii)] $F(x)$ has no root in $\mathbb{F}_{2^n}$ if and only if $b\ne 0$ and ${\rm Tr}_{1}^n(\frac{ac}{b^2})=1$.
\end{itemize}
\end{lemma}

\section{The Difference Distribution Table of $x^{2^{m+1}-1}$}\label{DDT-result1}

This section gives the explicit values of ${\rm DDT}_{F}(a,b)$ of the power mapping $F(x)=x^{2^{m+1}-1}$ over $\mathbb{F}_{2^n}$ for all $a,b \in \mathbb{F}_{2^n}$. The first main result is given by the following theorem, which is derived through the computation of the number of solutions over $\mathbb{F}_{2^n}$ of the equation  $x^{2^{m+1}-1}+(x+a)^{2^{m+1}-1}=b$.

\begin{theorem}\label{the-DDT}
Let $F(x)=x^{2^{m+1}-1}$ be a power mapping over $\mathbb{F}_{2^n}$. Then
\begin{itemize}
\item [\rm (i)] if $a=0$ and $b=0$, then ${\rm DDT}_{F}(a,b)=2^n$;
\item [\rm (ii)] if $a=0$ and $b\ne 0$, then ${\rm DDT}_{F}(a,b)=0$;
\item [\rm (iii)] if $a\ne 0$ and $b=0$, then
$${\rm DDT}_{F}(a,b)=0$$ for $n=2m+1$; and for $n=2m$,
$${\rm DDT}_{F}(a,b)=
\begin{cases}
    2, &  {\rm if}\,\, $m$\,\, {\rm is\,\, odd}; \\
     0, &  {\rm if}\,\,  $m$\,\, {\rm is\,\, even}. \\
\end{cases}
$$
\item [\rm (iv)] if $a\ne 0$, $b\ne 0$, let $t=\frac{b}{a^{2^{m+1}-1}}$, then
$${\rm DDT}_{F}(a,b)=
\begin{cases}
    2, &  {\rm if}\,\, {\rm Tr}_{1}^n(\frac{1}{t^{2^m+1}})=1; \\
     0, &  {\rm if}\,\,  {\rm Tr}_{1}^n(\frac{1}{t^{2^m+1}})=0 \\
\end{cases}
$$
for $n=2m+1$; and for $n=2m$,
$${\rm DDT}_{F}(a,b)=
\begin{cases}
    2^m, &  {\rm if}\,\,  t=1; \\
     2, &  {\rm if}\,\,  t^{2^m+1}\ne 1\,\, {\rm and}\,\, {\rm Tr}_{1}^m(\frac{t^{2^m}+t+1}{t^{2^{m+1}+2}+1})=1;\\
     0, &  {\rm if}\,\, t\ne 1\,\, {\rm and}\,\, t^{2^m+1}=1\, \\ &{\rm or}\,\, t^{2^m+1}\ne 1\,\, {\rm and}\,\, {\rm Tr}_{1}^m(\frac{t^{2^m}+t+1}{t^{2^{m+1}+2}+1})=0. \\

\end{cases}
$$
\end{itemize}
\end{theorem}

\begin{proof}
To determine the DDT of $F(x)$, it suffices to consider the number of solutions of
\begin{equation}\label{DDT-1}
F(x)+F(x+a)=x^{2^{m+1}-1}+(x+a)^{2^{m+1}-1}=b
\end{equation}
for any $a,b \in \mathbb{F}_{2^n}$. We divide the discussion into the following three cases.

 {\textbf{Case 1:}} It is seen that ${\rm DDT}_{F}(0,0)=2^n$ and ${\rm DDT}_{F}(0,b)=0$ for $b\ne 0$.

 {\textbf{Case 2:}} If $a\ne 0$ and $b=0$, (\ref{DDT-1}) can be reduced to
 \begin{equation}\label{DDT-b=0}
x^{2^{m+1}-1}=(x+a)^{2^{m+1}-1}.
\end{equation}
When $n=2m+1$, (\ref{DDT-b=0}) has no solution due to ${\rm gcd}(2^{m+1}-1,2^n-1)=1$.
When $n=2m$, due to
$${\rm gcd}(2^{m+1}-1,2^n-1)=
\begin{cases}
    3, &  {\rm if}\,\,  m\,\, {\rm is\,\, odd}; \\
     1, &  {\rm if}\,\,  m\,\, {\rm is\,\, even},
\end{cases}
$$
we have that (\ref{DDT-b=0}) has no solution and two solutions for even $m$ and odd $m$, respectively.

{\textbf{Case 3:}} In this case, we discuss the case $a\ne 0$ and $b\ne 0$. Let $y=\frac{x}{a}$, and $t=\frac{b}{a^{2^{m+1}-1}}$ where $t\ne 0$. Then (\ref{DDT-1}) can be reduced to $y^{2^{m+1}-1}+(y+1)^{2^{m+1}-1}=t$. Note that $y=0$ or $y=1$ is a solution of (\ref{DDT-1}) when $t=1$. Now we assume that $y\ne 0,1$.
Multiplying $y(y+1)$ on both sides of the above equation gives
 \begin{equation}\label{DDT-n=2m+1-1}
y^{2^{m+1}}+ty^2+(t+1)y=0.
\end{equation}

Next, we have to distinguish the following two subcases.

{\textbf{Subcase 1:}} $n=2m+1$.
When $t=1$, we derive from (\ref{DDT-n=2m+1-1}) that $y^{2^{m+1}}=y^2$, which implies $y\in \gf_{2^m}$. Since ${\rm gcd}(m,2m+1)=1$, then we have $y\in \gf_{2}$. This indicates ${\rm DDT}_{F}(a,b)=2$ for $t=1$. Now we suppose that $t\ne 0,1$. From (\ref{DDT-n=2m+1-1}), we obtain
\begin{equation}\label{DDT-n=2m+1-2}
\begin{aligned}
y^2&=(y^{2^{m+1}})^{2^{m+1}}\\
&=t^{2^{m+1}}y^{2^{m+2}}+(t^{2^{m+1}}+1)y^{2^{m+1}}\\
&=t^{2^{m+1}}(ty^2+(t+1)y)^2+(t^{2^{m+1}}+1)(ty^2+(t+1)y)\\
&=t^{2^{m+1}+2}y^4+(t^{{2^{m+1}}+2}+t^{2^{m+1}+1}+t^{2^{m+1}}+t)y^2+(t+1)(t^{2^{m+1}}+1)y.
\end{aligned}
\end{equation}
A straightforward calculation gives
$$
y(y+1)(t^{{2^{m+1}}+2}y^2+t^{{2^{m+1}}+2}y+(t^{2^{m+1}}+1)(t+1))=0.
$$
Since $y\ne 0,1$ and $t^{{2^{m+1}}+2}\ne 0$, then we have
\begin{equation}\label{DDT-n=2m+1-6}
y^2+y+\frac{(t^{2^{m+1}}+1)(t+1)}{t^{{2^{m+1}}+2}}=0.
\end{equation}
Thanks to
\begin{equation*}
\begin{aligned}
{\rm Tr}_{1}^n\Big (\frac{(t^{2^{m+1}}+1)(t+1)}{t^{{2^{m+1}}+2}}\Big)&={\rm Tr}_{1}^n\Big(\frac{1}{t}+\frac{1}{t^{2^{m+1}+1}}+\frac{1}{t^2}+\frac{1}{t^{2^{m+1}+2}}\Big)\\
&={\rm Tr}_{1}^n\Big(\frac{1}{t}\Big)+{\rm Tr}_{1}^n\Big(\frac{1}{t^{2^{m+1}+1}}\Big)+{\rm Tr}_{1}^n\Big(\frac{1}{t}\Big)^2+{\rm Tr}_{1}^n\Big(\frac{1}{t^{2^{m+1}+1}}\Big)^{2^m}\\
&=0
\end{aligned}
\end{equation*}
and Lemma \ref{lemma1-root}, (\ref{DDT-n=2m+1-6}) has two solutions.  We assume that $y_0$ is a solution of (\ref{DDT-n=2m+1-6}). Next, we determine whether $y_0$ is the solution of (\ref{DDT-n=2m+1-1}).
Combining (\ref{DDT-n=2m+1-1}) with (\ref{DDT-n=2m+1-6}), we get
$$
y_0^{2^{m+1}}+y_0+\frac{(t^{2^{m+1}}+1)(t+1)}{t^{{2^{m+1}}+1}}=0.
$$
On the other hand,
\begin{equation*}
\begin{aligned}
y_0^{2^{m+1}}+y_0&=\sum\limits_{i=0}^{m} (y_0^2+y_0)^{2^i}
=\sum\limits_{i=0}^{m} \Big(\frac{(t^{2^{m+1}}+1)(t+1)}{t^{{2^{m+1}}+2}}\Big)^{2^i}\\
&=\sum\limits_{i=0}^{m} \Big(\frac{1}{t}+\frac{1}{t^{2^{m+1}+1}}+\frac{1}{t^2}+\frac{1}{t^{2^{m+1}+2}}\Big)^{2^i}\\
&=\sum\limits_{i=0}^{m} \Big(\frac{1}{t}+(\frac{1}{t^{2^{m+1}+2}})^{2^m}+(\frac{1}{t})^2+\frac{1}{t^{2^{m+1}+2}}\Big)^{2^i}\\
&={\rm Tr}_{1}^n\Big(\frac{1}{t^{2^m+1}}\Big)^2+\frac{1}{t^{2^{m+1}+1}}+\frac{1}{t}+\frac{1}{t^{2^{m+1}}}\\
&={\rm Tr}_{1}^n\Big(\frac{1}{t^{2^m+1}}\Big)+1+\frac{(t^{2^{m+1}}+1)(t+1)}{t^{{2^{m+1}}+1}}.
\end{aligned}
\end{equation*}
Thus the solutions of (\ref{DDT-n=2m+1-6}) satisfy (\ref{DDT-n=2m+1-1}) if and only if ${\rm Tr}_{1}^n(\frac{1}{t^{2^m+1}})=1$. This implies that the number of the solutions of (\ref{DDT-n=2m+1-1}) is 2 (resp. 0) if and only if ${\rm Tr}_{1}^n(\frac{1}{t^{2^m+1}})=1$ (resp. ${\rm Tr}_{1}^n(\frac{1}{t^{2^m+1}})=0$ ) in this case, and the number of $t$ such that ${\rm Tr}_{1}^n(\frac{1}{t^{2^m+1}})=1$ is $2^{n-1}-1$.

{\textbf{Subcase 2:}} $n=2m$. When $t=1$, (\ref{DDT-n=2m+1-1}) can be reduced to $y^{2^{m+1}}=y^2$, which implies $y\in \gf_{2^m}$. Since ${\rm gcd}(m,2m)=m$, then we have $y\in \gf_{2^m}\setminus \{0,1\}$. From the discussion before Case 3,
 when $y=0,1$, we have $t=1$. Therefore (\ref{DDT-n=2m+1-1}) has $2^m$ solutions when $t=1$. Next, we suppose that $y\ne 0,1$ and $t\ne 0,1$. From (\ref{DDT-n=2m+1-1}),
we get
\begin{equation}\label{DDT-n=2m-1}
\begin{aligned}
y^4&=(y^{2^{m+1}})^{2^{m+1}}
=t^{2^{m+1}}y^{2^{m+2}}+(t^{2^{m+1}}+1)y^{2^{m+1}}\\
&=t^{2^{m+1}}(ty^2+(t+1)y)^2+(t^{2^{m+1}}+1)(ty^2+(t+1)y)\\
&=t^{2^{m+1}+2}y^4+(t^{{2^{m+1}}+2}+t^{2^{m+1}+1}+t^{2^{m+1}}+t)y^2+(t+1)(t^{2^{m+1}}+1)y.
\end{aligned}
\end{equation}
A direct calculation of (\ref{DDT-n=2m-1}) shows
$$
y(y+1)((t^{{2^{m+1}}+2}+1)y^2+(t^{{2^{m+1}}+2}+1)y+(t^{2^{m+1}}+1)(t+1))=0.
$$
Since $y\ne 0,1$, then we have
\begin{equation}\label{DDT-n=2m-5}
(t^{{2^{m+1}}+2}+1)y^2+(t^{{2^{m+1}}+2}+1)y+(t^{2^{m+1}}+1)(t+1)=0.
\end{equation}
Denote the unit circle of $\gf_{2^n}$ by $U=\{x\in \gf_{2^n}|x^{2^m+1}=1\}$. When $t\in U\setminus \{1\}$, (\ref{DDT-n=2m-5}) has no solution. Now we suppose that $t\notin U$. Then (\ref{DDT-n=2m-5}) is equivalent to
\begin{equation}\label{DDT-n=2m-6}
y^2+y+\frac{(t^{2^{m+1}}+1)(t+1)}{t^{{2^{m+1}}+2}+1}=0.
\end{equation}
Observe that
\begin{equation*}
\begin{aligned}
{\rm Tr}_{1}^n\Big(\frac{(t^{2^{m+1}}+1)(t+1)}{t^{{2^{m+1}}+2}+1}\Big)&={\rm Tr}_{1}^m\Big(\frac{(t+1)^{2^{m+1}+1}+(t+1)^{2^m+2}}{(t^{2^m+1}+1)^2}\Big)\\
&={\rm Tr}_{1}^m\Big(\frac{(t^{2^m}+t)(t^{2^m+1}+1)+(t^{2^m}+t)^2}{(t^{2^m+1}+1)^2}\Big)\\
&=0.
\end{aligned}
\end{equation*}
This, together with Lemma \ref{lemma1-root}, shows that (\ref{DDT-n=2m-6}) has two solutions. Assume that $y_0$ is a solution of (\ref{DDT-n=2m-6}), we then determine whether $y_0$ is the solution of (\ref{DDT-n=2m+1-1}).
Combining (\ref{DDT-n=2m+1-1}) with (\ref{DDT-n=2m-6}) gives
$$
y_0^{2^{m+1}}+y_0+\frac{(t^{2^{m+1}}+1)(t^2+t)}{t^{{2^{m+1}+2}}+1}=0.
$$
On the other hand,
\begin{equation*}
\begin{aligned}
&(y_0^{2^{m+1}}+y_0)^{2^m}=\Big(\sum\limits_{i=0}^{m} (y_0^2+y_0)^{2^i}\Big)^{2^m}=\sum\limits_{i=0}^{m} (y_0^2+y_0)^{2^{m+i}}\\
&=\sum\limits_{i=0}^{m} \Big(\frac{(t^{2^m}+1)(t^2+1)}{t^{{2^{m+1}}+2}+1}\Big)^{2^i}=\sum\limits_{i=0}^{m} \Big(\frac{t(t^{2^m+1}+1)+t^{2^m}+t^2+t+1}{t^{{2^{m+1}}+2}+1}\Big)^{2^i}\\
&=\sum\limits_{i=0}^{m} \Big(\frac{t}{t^{2^m+1}+1}+\frac{t^2}{(t^{2^m+1}+1)^2}+\frac{t^{2^m}}{(t^{2^m+1}+1)^2}+\frac{t}{(t^{2^m+1}+1)^2}+\frac{1}{(t^{2^m+1}+1)^2}\Big)^{2^i}\\
&=\frac{t}{t^{2^m+1}+1}+\frac{t^{2^{m+1}}}{(t^{2^m+1}+1)^2}+\sum\limits_{i=0}^{m} \Big(\Big(\frac{t}{(t^{2^m+1}+1)^2}\Big)^{2^m}+\frac{t}{(t^{2^m+1}+1)^2}+\frac{1}{(t^{2^m+1}+1)^2}\Big)^{2^i}\\
&=\frac{t^{2^{m+1}+2}+t^{2^{m+1}}+t^{2^m+2}+t^{2^m}}{(t^{2^m+1}+1)^2}+1+{\rm Tr}_{1}^n\Big(\frac{t}{(t^{2^m+1}+1)^2}\Big)+{\rm Tr}_{1}^m\Big(\frac{1}{(t^{2^m+1}+1)^2}\Big)\\
&=\Big(\frac{(t^{2^{m+1}}+1)(t^2+t)}{t^{{2^{m+1}+2}}+1}\Big)^{2^m}+1+{\rm Tr}_{1}^m\Big(\frac{t^{2^m}+t+1}{t^{2^{m+1}+2}+1}\Big).
\end{aligned}
\end{equation*}
Hence, the solutions of (\ref{DDT-n=2m-6}) satisfy (\ref{DDT-n=2m+1-1}) if and only if ${\rm Tr}_{1}^m(\frac{t^{2^m}+t+1}{t^{2^{m+1}+2}+1})=1$. Moreover, the number of the solutions of (\ref{DDT-n=2m+1-1}) is 2 (resp. 0) if and only if ${\rm Tr}_{1}^m(\frac{t^{2^m}+t+1}{t^{2^{m+1}+2}+1})=1$ (resp. ${\rm Tr}_{1}^m(\frac{t^{2^m}+t+1}{t^{2^{m+1}+2}+1})=0$) in this case.  This completes the proof.

\end{proof}

By Theorem \ref{the-DDT} and the identities in (\ref{prop}), we can directly derive the exact number of the pairs $(a,b)\in \gf_{2^n}^2$ having a given value of the DDT of $F(x)$ as follows.

\begin{proposition}\label{the-FBCT-0}
Let $m$ be a positive integer and $F(x)=x^{2^{m+1}-1}$ be a power mapping over $\gf_{2^n}$. The number of pairs $(a,b)\in \gf_{2^n}^2$ in the spectrum of the {\rm DDT} of $F(x)$ satisfies
\begin{itemize}
\item [\rm (i)] $|\{(a,b)\in \gf_{2^n}^2: {\rm DDT}_{F}(a,b)=2^n\}|=1$;
\item [\rm (ii)] if $n=2m+1$, then
\begin{equation*}
\begin{aligned}
&|\{(a,b)\in \gf_{2^n}^2: {\rm DDT}_{F}(a,b)=2\}|=2^{n-1}(2^n-1);\\
&|\{(a,b)\in \gf_{2^n}^2: {\rm DDT}_{F}(a,b)=0\}|=2^{2n-1}+2^{n-1}-1;
\end{aligned}
\end{equation*}
\item [\rm (iii)] if $n=2m$, then
\begin{equation*}
\begin{aligned}
&|\{(a,b)\in \gf_{2^n}^2: {\rm DDT}_{F}(a,b)=2^m\}|=2^n-1;\\
&|\{(a,b)\in \gf_{2^n}^2: {\rm DDT}_{F}(a,b)=2\}|=(2^{n-1}-2^{m-1})(2^n-1);\\
&|\{(a,b)\in \gf_{2^n}^2: {\rm DDT}_{F}(a,b)=0\}|=(2^{n-1}+2^{m-1})(2^n-1).
\end{aligned}
\end{equation*}
\end{itemize}
\end{proposition}

\section{The Feistel Boomerang Connectivity Table of $x^{2^{m+1}-1}$}\label{FBCT-result1}

This section is devoted to presenting a detailed study of the FBCT of the power mapping $F(x)=x^{2^{m+1}-1}$ over $\gf_{2^n}$, where $n=2m+1$ or $n=2m$. The second main result is given by the following theorem, which is derived through the computation of the number of solutions over $\mathbb{F}_{2^n}$ of the equation $x^{2^{m+1}-1}+(x+a)^{2^{m+1}-1}+(x+b)^{2^{m+1}-1}+(x+a+b)^{2^{m+1}-1}=0$.

\begin{theorem}\label{the-FBCT-1}
Let $F(x)=x^{2^{m+1}-1}$ be a power mapping over $\gf_{2^n}$. Then
\begin{itemize}
\item [\rm (i)] if $n=2m$, then
$$
{\rm FBCT}_{F}(a,b)=
\begin{cases}
    2^n, &  {\rm if}\,\,  ab(a+b); \\
    2^m, &  {\rm if}\,\, ab(a+b)\ne 0\,\, {\rm and}\,\,\frac{a}{b}\in \gf_{2^m} \setminus \{0,1\};\\
     0, &  {\rm if}\,\, ab(a+b)\ne 0\,\, {\rm and}\,\, \frac{a}{b}\in \gf_{2^n} \setminus \gf_{2^m};\\
\end{cases}
$$
\item [\rm (ii)] if $n=2m+1$, then
$$
{\rm FBCT}_{F}(a,b)=
\begin{cases}
    2^n, &  {\rm if}\,\,  ab(a+b); \\
     0, &  {\rm if}\,\, ab(a+b)\ne 0. \\
\end{cases}
$$
\end{itemize}
\end{theorem}

\begin{proof}
To prove this theorem, according to Definition \ref{definition-FBCT}, we need to
count the number of the solutions of
$$F(x)+F(x+a)+F(x+b)+F(x+a+b)=0,$$
i.e.,
\begin{equation}\label{FBCT-n=2m-1}
x^{2^{m+1}-1}+(x+a)^{2^{m+1}-1}+(x+b)^{2^{m+1}-1}+(x+a+b)^{2^{m+1}-1}=0,
\end{equation}
where $a,b\in \gf_{2^n}$.

We start by considering the case $n=2m$.

{\textbf{Case 1:}} When $a=0$ or $b=0$ or $a=b$ with $a\ne 0$, it can be easily seen that (\ref{FBCT-n=2m-1}) holds for all $x\in \gf_{2^n}$, which gives
$$
 {\rm FBCT}_{F}(a,b)=2^n.
$$

{\textbf{Case 2:}} Assume that $ab(a+b)\ne 0$. If $x=0$, $a$, $b$ or $a+b$, then (\ref{FBCT-n=2m-1}) becomes
\begin{equation}\label{FBCT-n=2m-2}
a^{2^{m+1}-1}+b^{2^{m+1}-1}+(a+b)^{2^{m+1}-1}=0.
\end{equation}
Multiplying $ab(a+b)$ on both sides of (\ref{FBCT-n=2m-2}), it can be further simplified as $(\frac{a}{b})^{2^{m+1}}=(\frac{a}{b})^2$. Then we have $\frac{a}{b}\in \gf_{2^m}\setminus \{0,1\}$.  Next, assume that $x\ne 0$, $x\ne a$, $x\ne b$ and $x\ne a+b$. Let $c=\frac{a}{b}$ and $y=\frac{x}{b}$, we have $c\ne 0,1$ and $y\ne 0,1,c,c+1$.  Then, (\ref{FBCT-n=2m-1}) is equivalent to
$$
b^{2^{m+1}-1}(y^{2^{m+1}-1}+(y+c)^{2^{m+1}-1}+(y+1)^{2^{m+1}-1}+(y+c+1)^{2^{m+1}-1})=0.
$$
Since $b\ne 0$, thus we only need to consider the solutions of
\begin{equation}\label{FBCT-n=2m-3}
y^{2^{m+1}-1}+(y+c)^{2^{m+1}-1}+(y+1)^{2^{m+1}-1}+(y+c+1)^{2^{m+1}-1}=0.
\end{equation}
Multiplying $y(y+c)(y+1)(y+c+1)$ on both sides of (\ref{FBCT-n=2m-3}) gives
 \begin{equation*}
\begin{aligned}
&y^{2^{m+1}}(y+c)(y+1)(y+c+1)+(y^{2^{m+1}}+c^{2^{m+1}})y(y+1)(y+c+1)\\
&+(y^{2^{m+1}}+1)y(y+c)(y+c+1)+(y^{2^{m+1}}+c^{2^{m+1}}+1)y(y+c)(y+1)=0.
\end{aligned}
\end{equation*}
Expanding each of the terms of the above equation leads to
\begin{equation}\label{FBCT-n=2m-4}
y^{2^{m+1}}(c^2+c)+y^2(c^{2^{m+1}}+c)+y(c^{2^{m+1}}+c^2)=0.
\end{equation}
Since $c\ne 0,1$, (\ref{FBCT-n=2m-4}) is equivalent to
\begin{equation}\label{FBCT-n=2m-5}
y^{2^{m+1}}=\frac{c^{2^{m+1}}+c}{c^2+c}y^2+\frac{c^{2^{m+1}}+c^2}{c^2+c}y.
\end{equation}

{\textbf{Subcase 1:}}
Assume that $c\in \gf_{2^m}\setminus \{0,1\}$. Then $c^{2^m}=c$ and $c^{2^{m+1}}=c^2$. (\ref {FBCT-n=2m-5}) reduces to $(c^2+c)(y^{2^{m+1}}+y^2)=0$. We have $y\in \gf_{2^m}\setminus \{0,1,c,c+1\}$ since $c\ne 0,1$ and ${\rm gcd}(n,m)=m$. According to the discussion at the beginning of Case 2, we get $y=0$, $y=1$, $y=c$ and $y=c+1$ are the solutions of (\ref{FBCT-n=2m-4}) if $c\in \gf_{2^m}\setminus \{0,1\}$. Therefore (\ref{FBCT-n=2m-4}) has $2^m$ solutions when $c\in \gf_{2^m}\setminus \{0,1\}$.

{\textbf{Subcase 2:}}
Assume that $c\in \gf_{2^n}\setminus \gf_{2^m}$.
Raising $2^{m+1}$-th power to (\ref {FBCT-n=2m-5}), we get
\begin{equation}\label{FBCT-n=2m-6}
y^4=\frac{c^{2^{m+1}}+c^4}{(c^2+c)^{2^{m+1}}}y^{2^{m+2}}+\frac{c^{2^{m+2}}+c^4}{(c^2+c)^{2^{m+1}}}y^{2^{m+1}}.
\end{equation}
Substituting (\ref{FBCT-n=2m-5}) into (\ref{FBCT-n=2m-6}) and then multiplying $(c^2+c)^{2^{m+1}+2}$ on both sides of this equation, we have
\begin{equation*}%\label{FBCT-n=2m-7}
(c^{2^{m+1}}+c^2)^2((c^{2^{m+1}}+c^2)y^4+(c^{2^{m+1}}+c^2)(c^2+c+1)y^2+(c^{2^{m+1}}+c^2)(c^2+c)y)=0,
\end{equation*}
which can be further rewritten as
\begin{equation}\label{FBCT-n=2m-8}
(c^{2^{m+1}}+c^2)^3y(y+1)(y+c)(y+c+1)=0.
\end{equation}
Since $c\ne 0,1$ and $y\ne 0,1,c,c+1$, we have (\ref{FBCT-n=2m-8}) has no solution. Hence, (\ref{FBCT-n=2m-5}) has no solution in this case. This completes the proof of the case for $n=2m$.

For the case of $n=2m+1$, the desired result can be similarly proved. Note that in this case ${\rm gcd}(n,m)={\rm gcd}(2m+1,m)=1$ which leads to (\ref{FBCT-n=2m-4}) has no solution. This implies that ${\rm FBCT}_{F}(a,b)=2^m$ can not be achieved. This completes the proof.
\end{proof}
The distribution of the values of the FBCT of the power mapping $F(x)=x^{2^{m+1}-1}$ over $\gf_{2^n}$  can be obtained according to Theorem \ref{the-FBCT-1}. We omit the proof here.

\begin{proposition}\label{the-spec-FBCT-1}
Let $m$ be a positive integer and $F(x)=x^{2^{m+1}-1}$ be a power mapping over $\gf_{2^n}$.
\begin{itemize}
\item [\rm (i)] If $n=2m$, then
\begin{equation*}
\begin{aligned}
&|\{(a,b)\in \gf_{2^n}^2:{\rm FBCT}_{F}(a,b)=2^n\}|=3\times 2^n-2;\\
&|\{(a,b)\in \gf_{2^n}^2:{\rm FBCT}_{F}(a,b)=2^m\}|=(2^m-2)(2^n-1);\\
&|\{(a,b)\in \gf_{2^n}^2:{\rm FBCT}_{F}(a,b)=0\}|=(2^n-2^m)(2^n-1);
\end{aligned}
\end{equation*}
\item [\rm (ii)] If $n=2m+1$, then
\begin{equation*}
\begin{aligned}
&|\{(a,b)\in \gf_{2^n}^2:{\rm FBCT}_{F}(a,b)=2^n\}|=3\times 2^n-2;\\
&|\{(a,b)\in \gf_{2^n}^2:{\rm FBCT}_{F}(a,b)=0\}|=(2^n-2)(2^n-1).
\end{aligned}
\end{equation*}
\end{itemize}
\end{proposition}

As a consequence,  the Feistel boomerang uniformity of  $F(x)=x^{2^{m+1}-1}$ over $\gf_{2^n}$.

\begin{corollary}\label{the-spec-FBCT-2}
Let $m$ be a positive integer and $F(x)=x^{2^{m+1}-1}$ be a power mapping over $\gf_{2^n}$. Then the Feistel boomerang uniformity of $F(x)$ satisfies
$$\beta_c(F)=
\begin{cases}
    2^m, &  {\rm if}\,\,  n=2m; \\
    0, &  {\rm if}\,\, n=2m+1.\\
\end{cases}
$$
\end{corollary}

\section{The Feistel Boomerang Difference Table of $x^{2^{m+1}-1}$}\label{FBDT-result1}

In this section, we give explicit values of all the entries of the FBDT of $F(x)=x^{2^{m+1}-1}$ over $\gf_{2^n}$, where $n=2m+1$ or $n=2m$.
The third main result is given by the following theorem, which is derived through the computation of the number of solutions over $\mathbb{F}_{2^n}$ of the system of equations:
$$
\left \{\begin{array}{ll}
x^{2^{m+1}-1}+(x+a)^{2^{m+1}-1}+(x+b)^{2^{m+1}-1}+(x+a+b)^{2^{m+1}-1}=0,\\
x^{2^{m+1}-1}+(x+a)^{2^{m+1}-1}=c.
\end{array}\right.
$$

\begin{theorem}\label{the-FBDT-1}
Let $F(x)=x^{2^{m+1}-1}$ be a power mapping over $\gf_{2^n}$. Then
\begin{itemize}
\item [\rm (i)] if $n=2m+1$, then
$${\rm FBDT}_{F}(a,c,b)=
\begin{cases}
    2^n, &  {\rm if}\,\,  (a,c,b)\in I_{2^n}; \\
     2, &  {\rm if}\,\, (a,c,b)\in I_{2}; \\
     0, &  {\rm if}\,\, (a,c,b)\in I_{0}^{(1)}\cup I_{0}^{(2)}\cup I_{0}^{(3)},
\end{cases}
$$
where $t=\frac{c}{a^{2^{m+1}-1}}$ and
%\begin{equation*}
\begin{align*}
&I_{2^n}=\{(a,c,b)\in \gf_{2^n}^3: a=c=0\},\\
&I_{2}=\{(a,c,b)\in \gf_{2^n}^3: a\ne 0,\, b\in \{0, a\}\,\, {\rm and}\,\, {\rm Tr}_{1}^n\Big(\frac{1}{t^{2^m+1}}\Big)=1\},\\
&I_{0}^{(1)}=\{(a,c,b)\in \gf_{2^n}^3: a=0\,\, {\rm and}\,\, c\ne 0\},\\
&I_{0}^{(2)}=\{(a,c,b)\in \gf_{2^n}^3: a\ne 0,\, b\in \{0, a\}\,\,  {\rm and}\,\, {\rm Tr}_{1}^n\Big(\frac{1}{t^{2^m+1}}\Big)=0\},\\
&I_{0}^{(3)}=\{(a,c,b)\in \gf_{2^n}^3: ab(a+b)\ne 0 \},
\end{align*}
\item [\rm (ii)] if $n=2m$, then
$${\rm FBDT}_{F}(a,c,b)=
\begin{cases}
    2^n, &  {\rm if}\,\,  (a,c,b)\in I_{2^n}; \\
    2^m, &  {\rm if}\,\, (a,c,b)\in I_{2^m}^{(1)}\cup I_{2^m}^{(2)}; \\
     2, &  {\rm if}\,\, (a,c,b)\in I_{2}^{(1)}\cup I_{2}^{(2)}; \\
     0, &  {\rm if}\,\, (a,c,b)\in I_{0}^{(1)}\cup I_{0}^{(2)}\cup I_{0}^{(3)}\cup I_{0}^{(4)}\cup I_{0}^{(5)}\cup I_{0}^{(6)},
\end{cases}
$$
where $t=\frac{c}{a^{2^{m+1}-1}}$ and
%\begin{equation*}
\begin{align*}
%\begin{aligned}
&I_{2^n}=\{(a,c,b)\in \gf_{2^n}^3: a=c=0\},\\
&I_{2^m}^{(1)}=\{(a,c,b)\in \gf_{2^n}^3: a\ne 0,\, c\ne 0,\, b\in \{0, a\}\,\, {\rm and}\,\, t=1\},\\
&I_{2^m}^{(2)}=\{(a,c,b)\in \gf_{2^n}^3: a\ne 0,\,c\ne 0,\, b\notin \{0, a\},\, t=1\,\, {\rm and}\,\, \frac{a}{b}\in \gf_{2^m}\setminus \{0,1\}\},\\
&I_{2}^{(1)}=\{(a,c,b)\in \gf_{2^n}^3: a\ne 0,\, c\ne 0,\, b\in \{0, a\},\,  t^{2^m+1}\ne 1\,\, {\rm and}\,\, {\rm Tr}_{1}^m\Big(\frac{t^{2^m}+t+1}{t^{2^{m+1}+2}+1}\Big)=1\},\\
&I_{2}^{(2)}=\{(a,c,b)\in \gf_{2^n}^3: a\ne 0,\, c=0,\, b\in \{0, a\}\,\, {\rm and}\,\, m\,\, {\rm is\,\, odd}\},\\
&I_{0}^{(1)}=\{(a,c,b)\in \gf_{2^n}^3: a=0\,\, {\rm and}\,\, c\ne 0\},\\
&I_{0}^{(2)}=\{(a,c,b)\in \gf_{2^n}^3: a\ne 0,\, c=0,\, b\in \{0, a\}\,\, {\rm and}\,\, m\,\, {\rm is\,\, even}\},\\
&I_{0}^{(3)}=\{(a,c,b)\in \gf_{2^n}^3: a\ne 0,\, c\ne 0,\, b\in \{0, a\},\,  t^{2^m+1}\ne 1\,\, {\rm and}\,\, {\rm Tr}_{1}^m\Big(\frac{t^{2^m}+t+1}{t^{2^{m+1}+2}+1}\Big)=0\},\\
&I_{0}^{(4)}=\{(a,c,b)\in \gf_{2^n}^3: a\ne 0,\, c\ne 0,\, b\in \{0, a\},\, t\ne 1\,\, {\rm and}\,\, t^{2^m+1}=1\},\\
&I_{0}^{(5)}=\{(a,c,b)\in \gf_{2^n}^3: a\ne 0,\, c\ne 0,\,  b\notin \{0, a\},\,t\ne 1\,\,  {\rm and}\,\, \frac{a}{b}\in \gf_{2^m}\setminus \{0,1\}\},\\
&I_{0}^{(6)}=\{(a,c,b)\in \gf_{2^n}^3: a\ne 0,\, c\ne 0,\, b\notin \{0, a\}\,\, {\rm and}\,\, \frac{a}{b}\in \gf_{2^n}\setminus \gf_{2^m}\}.
%\end{aligned}
\end{align*}
%\end{equation*}
\end{itemize}
\end{theorem}
\begin{proof}
Let $a, c, b\in \gf_{2^n}$, we need to determine the number of the solutions of
\begin{equation}\label{FBDT-n=2m-1-1}
\arraycolsep=1.2pt\def\arraystretch{1.4}
\left \{\begin{array}{ll}
F(x)+F(x+a)+F(x+b)+F(x+a+b)=0,\\
F(x)+F(x+a)=c.
\end{array}\right.
\end{equation}

First, we consider the case of $n=2m+1$.

{\textbf{Case 1:}} Assume that $a=0$. Obviously, (\ref{FBDT-n=2m-1-1}) has $0$
or $2^n$ solutions and it has $2^n$ solutions iff $c=0$. This implies that
$$
{\rm FBDT}_{F}(0,c,b)=
\begin{cases}
    2^n, &  {\rm if}\,\, a=0,\, c=0; \\
     0, &  {\rm if}\,\, a=0,\, c\ne 0.
\end{cases}
$$

{\textbf{Case 2:}} Assume that $a\ne 0$. If $b=0$ or $b=a$,  (\ref{FBDT-n=2m-1-1}) is reduced to $F(x)+F(x+a)=c$. Hence, we have ${\rm FBDT}_{F}(a,c,b)={\rm DDT}_{F}(a,c)$. Let $t=\frac{c}{a^{2^{m+1}-1}}$, then by Theorem \ref{the-DDT}, we have
$$
{\rm FBDT}_{F}(a, c, b)=
\begin{cases}
    2, &  {\rm if}\,\, a\ne 0,\, c\ne 0,\,  b\in \{0, a\}\,\,{\rm and}\,\, {\rm Tr}_{1}^n\Big(\frac{1}{t^{2^m+1}}\Big)=1;\\
    0, &  {\rm if}\,\,  a\ne 0,\, c=0,\, b\in \{0, a\},\,  \\
    &{\rm or}\,\, a\ne 0,\, c\ne0,\, b\in \{0, a\}\,\,  {\rm and}\,\, {\rm Tr}_{1}^n\Big(\frac{1}{t^{2^m+1}}\Big)=0.
\end{cases}
$$

{\textbf{Case 3:}} Assume here that $a\ne 0$, $b\ne 0$ and $a\ne b$. According to Theorem \ref{the-FBCT-1}, the first equation of (\ref{FBDT-n=2m-1-1}) has no solution. Hence, (\ref{FBDT-n=2m-1-1}) has no solution.

Next, we study the case when $n=2m$.

{\textbf{Case 1:}} Assume that $a=0$. The proof is similar to the discussion in Case 1 for $n=2m+1$.

{\textbf{Case 2:}}
Assume that $a\ne 0$. If $b=0$ or $b=a$, similar to the discussion in Case 2 for $n=2m+1$, we also have ${\rm FBDT}_{F}(a,c,b)={\rm DDT}_{F}(a,c)$ in this case. Let $t=\frac{c}{a^{2^{m+1}-1}}$, then by applying Theorem \ref{the-DDT}, we get
$$
{\rm FBDT}_{F}(a, c, b)=
\begin{cases}
    2^m, &  {\rm if}\,\, a\ne 0,\, c\ne 0,\, b\in \{0, a\},\, t=1\, ;\\
    2, &  {\rm if}\,\, a\ne 0,\, c\ne 0,\, b\in \{0, a\},\, t^{2^m+1}\ne 1\,\, {\rm and}\,\, {\rm Tr}_{1}^m\Big(\frac{t^{2^m}+t+1}{t^{2^{m+1}+2}+1}\Big)=1\\
    &{\rm or}\,\, a\ne 0,\, c=0,\, b\in \{0, a\}\,\, {\rm and}\,\, m\,\, {\rm is\,\, odd}; \\
    0, &  {\rm if}\,\, a\ne 0,\, c=0,\, b\in \{0, a\}\,\, {\rm and}\,\, m\,\, {\rm is\,\, even} \\
   &{\rm or}\,\,  a\ne 0,\, c\ne 0,\, b\in \{0, a\},\, t\ne 1\,\, {\rm and}\,\, t^{2^m+1}=1\, \\
   &{\rm or}\,\, a\ne 0,\, c\ne 0,\, b\in \{0, a\},\, t^{2^m+1}\ne 1\,\, {\rm and}\,\, {\rm Tr}_{1}^m\Big(\frac{t^{2^m}+t+1}{t^{2^{m+1}+2}+1}\Big)=0.
\end{cases}
$$

{\textbf{Case 3:}}
Assume that $ab(a+b)\ne 0$. According to Case 2 in Theorem \ref{the-FBCT-1}, we have $x=0, a, b, a+b$ are the solutions of the first equation of (\ref{FBDT-n=2m-1-1}) when $\frac{a}{b}\in \gf_{2^m}\setminus \{0,1\}$. Note that $x=0, a, b, a+b$ are the solutions of the second equation of (\ref{FBDT-n=2m-1-1}) if $t=1$. Thus $x=0, a, b, a+b$ are the solutions of (\ref{FBDT-n=2m-1-1}) when $\frac{a}{b}\in \gf_{2^m}\setminus \{0,1\}$ and $t=1$. Next we assume that $x\notin \{0, a, b, a+b\}$. If $\frac{a}{b}\in \gf_{2^m}\setminus \{0,1\}$ and $t=1$, we get $a^2b^{2^{m+1}}+a^{2^{m+1}}b^2=0$ and $ca=a^{2^{m+1}}$.  Multiplying $x(x+a)(x+b)(x+a+b)$ and $x(x+a)$ on both sides of the first and the second equations of (\ref{FBDT-n=2m-1-1}) respectively gives
\begin{equation}\label{FBDT-n=2m-1-2}
\arraycolsep=1.2pt\def\arraystretch{1.4}
\left \{\begin{array}{ll}
a^2bx^{2^{m+1}}+ab^2x^{2^{m+1}}+a^{2^{m+1}}bx^2+ab^{2^{m+1}}x^2=0,\\
a^2x^{2^{m+1}}+a^{2^{m+1}}x^2=0.
\end{array}\right.
\end{equation}
Since $a,b\ne 0$, $a^{2^{m+1}}x^2=a^2x^{2^{m+1}}$ and $b^{2^{m+1}}=\frac{a^{2^{m+1}}b^2}{a^2}$, then (\ref{FBDT-n=2m-1-2}) can be further reduced to
$$
a^2x^{2^{m+1}}+a^{2^{m+1}}x^2=0,
$$
which has  $2^m$ solutions due to ${\gcd}(2^{m+1}-2,2^n-1)=2^m-1$, including the solutions $x=0, a, b, a+b$.
Thus (\ref{FBDT-n=2m-1-2}) has $2^m$ solutions when $\frac{a}{b}\in \gf_{2^m}\setminus \{0,1\}$ and $t=1$. According to Theorem \ref{the-FBCT-1}, when $ab(a+b)\ne 0$ and $\frac{a}{b}\in \gf_{2^n}\setminus \gf_{2^m}$, the first equation of (\ref{FBDT-n=2m-1-1}) has no solution.
To conclude this case, we obtain
$$
{\rm FBDT}_{F}(a, c, b)=
\begin{cases}
    2^m, &  {\rm if}\,\, a\ne 0,\, c\ne 0,\, b\notin \{0,a\},\, \frac{a}{b}\in \gf_{2^m}\setminus \{0,1\}\,\, {\rm and}\,\, t=1\,;\\
    0, &  {\rm if}\,\, a\ne 0,\, c\ne 0,\, b\notin \{0,a\},\,  \frac{a}{b}\in \gf_{2^m}\setminus \{0,1\}\,\, {\rm and}\,\, t\ne 1\, \\
    &{\rm or}\,\, a\ne 0,\, c\ne 0,\, b\notin \{0,a\},\, \frac{a}{b}\in \gf_{2^n}\setminus \gf_{2^m}.
\end{cases}
$$
This completes the proof.
\end{proof}

Using Theorem \ref{the-DDT} and Theorem \ref{the-FBDT-1}, the following result can be derived.

\begin{proposition}\label{cou-FBDT-1}
Let $m$ be a positive integer and $F(x)=x^{2^{m+1}-1}$ be a power mapping over $\gf_{2^n}$.
\begin{itemize}
\item [\rm (i)] If $n=2m+1$, then
\begin{equation*}
\begin{aligned}
&|\{(a,c,b)\in \gf_{2^n}^3:{\rm FBDT}_{F}(a,c,b)=2^n\}|=2^n;\\
&|\{(a,c,b)\in \gf_{2^n}^3:{\rm FBDT}_{F}(a,c,b)=2\}|=2^n(2^n-1);\\
&|\{(a,c,b)\in \gf_{2^n}^3:{\rm FBDT}_{F}(a,c,b)=0\}|=2^{2n}(2^n-1);
\end{aligned}
\end{equation*}

\item [\rm (ii)] If $n=2m$, then
\begin{equation*}
\begin{aligned}
&|\{(a,c,b)\in \gf_{2^n}^3:{\rm FBDT}_{F}(a,c,b)=2^n\}|=2^n;\\
&|\{(a,c,b)\in \gf_{2^n}^3:{\rm FBDT}_{F}(a,c,b)=2^m\}|=2^m(2^n-1);\\
&|\{(a,c,b)\in \gf_{2^n}^3:{\rm FBDT}_{F}(a,c,b)=2\}|=(2^n-2^m)(2^n-1);\\
&|\{(a,c,b)\in \gf_{2^n}^3:{\rm FBDT}_{F}(a,c,b)=0\}|=2^{2n}(2^n-1).
\end{aligned}
\end{equation*}
\end{itemize}

\end{proposition}

Consequently, the  Feistel boomerang differential uniformity of the power mapping $F(x)=x^{2^{m+1}-1}$ over $\gf_{2^n}$ can be determined as below.

\begin{corollary}\label{the-spec-FBDT-2}
Let $m$ be a positive integer and  $F(x)=x^{2^{m+1}-1}$ be a power mapping over $\gf_{2^n}$. Then the Feistel boomerang differential uniformity of $F(x)$ satisfies
$$
\beta_d(F)=
\begin{cases}
    2^m, &  {\rm if}\,\,  n=2m; \\
    2, &  {\rm if}\,\, n=2m+1.\\
\end{cases}
$$
\end{corollary}

\section{The Feistel Boomerang Extended Table of $x^{2^{m+1}-1}$}\label{FBET-result1}

This section determines the entries of the FBET of the power mapping $F(x)=x^{2^{m+1}-1}$ over $\gf_{2^n}$, where $n=2m+1$ or $n=2m$. The fourth main result is given by the following theorem, which is derived through the computation of the number of solutions over $\mathbb{F}_{2^n}$ of the system of equations:
$$
\left \{\begin{array}{ll}
x^{2^{m+1}-1}+(x+a)^{2^{m+1}-1}+(x+b)^{2^{m+1}-1}+(x+a+b)^{2^{m+1}-1}=0,\\
x^{2^{m+1}-1}+(x+a)^{2^{m+1}-1}=c,\\
(x+a)^{2^{m+1}-1}+(x+a+b)^{2^{m+1}-1}=d.
\end{array}\right.
$$

\begin{theorem}\label{the-FBET-1}
Let $F(x)=x^{2^{m+1}-1}$ be a power mapping over $\gf_{2^n}$. Then
\begin{itemize}
\item [\rm (i)] if $n=2m+1$, then
$$
{\rm FBET}_{F}(a,c,b,d)=
\begin{cases}
    2^n, &  {\rm if}\,\,  (a,c,b,d)\in I_{2^n}; \\
     2, &  {\rm if}\,\, (a,c,b,d)\in I_{2}^{(1)}\cup I_{2}^{(2)}\cup I_{2}^{(3)};\\
     0, &  {\rm ortherwise},
\end{cases}
$$
where $t_1=\frac{c}{a^{2^{m+1}-1}}$, $t_2=\frac{d}{b^{2^{m+1}-1}}$, and
\begin{equation*}
\begin{aligned}
&I_{2^n}=\{(a,c,b,d)\in \gf_{2^n}^4: a=c=b=d=0\},\\
&I_{2}^{(1)}=\{(a,c,b,d)\in \gf_{2^n}^4: a=c=0,\, bd\ne 0\,\, {\rm and}\,\, {\rm Tr}_{1}^n\Big(\frac{1}{t_2^{2^m+1}}\Big)=1\},\\
&I_{2}^{(2)}=\{(a,c,b,d)\in \gf_{2^n}^4: a\ne 0,\, c\ne 0,\, b=d=0\,\, {\rm and}\,\, {\rm Tr}_{1}^n\Big(\frac{1}{t_1^{2^m+1}}\Big)=1\},\\
&I_{2}^{(3)}=\{(a,c,b,d)\in \gf_{2^n}^4: a=b\ne 0,\, c=d\ne 0\,\, {\rm and}\,\, {\rm Tr}_{1}^n\Big(\frac{1}{t_1^{2^m+1}}\Big)=1\}.\\
\end{aligned}
\end{equation*}

\item [\rm (ii)] if $n=2m$, then
$${\rm FBET}_{F}(a,c,b,d)=
\begin{cases}
    2^n, &  {\rm if}\,\,  (a,c,b,d)\in I_{2^n}; \\
    2^m, &  {\rm if}\,\, (a,c,b,d)\in I_{2^m}^{(1)}\cup I_{2^m}^{(2)}\cup I_{2^m}^{(3)}\cup I_{2^m}^{(4)}; \\
     2, &  {\rm if}\,\, (a,c,b,d)\in I_{2}^{(1)}\cup I_{2}^{(2)}\cup I_{2}^{(3)}\cup I_{2}^{(4)}\cup I_{2}^{(5)}\cup I_{2}^{(6)}; \\
     0, &  {\rm ortherwise},
\end{cases}
$$
where $t_1=\frac{c}{a^{2^{m+1}-1}}$, $t_2=\frac{d}{b^{2^{m+1}-1}}$ and
%\begin{equation*}
\begin{align*}
&I_{2^n}=\{(a,c,b,d)\in \gf_{2^n}^4: a=c=b=d=0\},\\
&I_{2^m}^{(1)}=\{(a,c,b,d)\in \gf_{2^n}^4: a=c=0,\, bd\ne 0\,\, {\rm and}\,\, t_2=1\},\\
&I_{2^m}^{(2)}=\{(a,c,b,d)\in \gf_{2^n}^4: a\ne 0,\,c\ne 0,\, b=d=0\,\, {\rm and}\,\, t_1=1\},\\
&I_{2^m}^{(3)}=\{(a,c,b,d)\in \gf_{2^n}^4: a=b\ne 0,\,c=d\ne 0\,\,  {\rm and}\,\, t_1=1\},\\
&I_{2^m}^{(4)}=\{(a,c,b,d)\in \gf_{2^n}^4: a\ne 0,\, b\notin \{0,a\},\, \frac{a}{b}\in \gf_{2^m}\setminus \{0,1\},\, t_1=1\,\, {\rm and}\,\, t_2=1 \},\\
&I_{2}^{(1)}=\{(a,c,b,d)\in \gf_{2^n}^4: a=c=0,\,  bd\ne 0,\,  t_2^{2^m+1}\ne 1\,\, {\rm and}\,\, {\rm Tr}_{1}^m\Big(\frac{t_2^{2^m}+t_2+1}{t_2^{2^{m+1}+2}+1}\Big)=1\},\\
&I_{2}^{(2)}=\{(a,c,b,d)\in \gf_{2^n}^4: a=c=0,\,  b\ne 0,\,  d=0\,\, {\rm and}\,\, m\,\, {\rm is\,\, odd}\},\\
&I_{2}^{(3)}=\{(a,c,b,d)\in \gf_{2^n}^4: a\ne 0,\, c\ne 0,\, b=d=0,\, t_1^{2^m+1}\ne 1\,\, {\rm and}\,\, {\rm Tr}_{1}^m\Big(\frac{t_1^{2^m}+t_1+1}{t_1^{2^{m+1}+2}+1}\Big)=1\},\\
&I_{2}^{(4)}=\{(a,c,b,d)\in \gf_{2^n}^4: a\ne 0,\, c=0,\, b=d=0\,\,  {\rm and}\,\, m\,\, {\rm is\,\, odd}\},\\
&I_{2}^{(5)}=\{(a,c,b,d)\in \gf_{2^n}^4: a=b\ne 0,\, c=d\ne 0,\,  t_1^{2^m+1}\ne 1\,\, {\rm and}\,\, {\rm Tr}_{1}^m\Big(\frac{t_1^{2^m}+t_1+1}{t_1^{2^{m+1}+2}+1}\Big)=1\}.\\
&I_{2}^{(6)}=\{(a,c,b,d)\in \gf_{2^n}^4: a=b\ne 0,\, c=d=0\,\,  {\rm and}\,\, m\,\, {\rm is\,\, odd}\}.
\end{align*}
%\end{equation*}
\end{itemize}
\end{theorem}

\begin{proof}
Let $a, c, b, d\in \gf_{2^n}$, we need to count the number of the solutions of
\begin{equation}\label{FBET-n=2m-1-1}
\arraycolsep=1.2pt\def\arraystretch{1.4}
\left \{\begin{array}{ll}
F(x)+F(x+a)+F(x+b)+F(x+a+b)=0,\\
F(x)+F(x+a)=c,\\
F(x+a)+F(x+a+b)=d.
\end{array}\right.
\end{equation}

We first consider the case  $n=2m+1$ as follows:

{\textbf{Case 1:}} It can be seen that ${\rm FBET}(0,0,0,0)=2^n$.

{\textbf{Case 2:}} Assume that $a=0$. The second equation of (\ref{FBET-n=2m-1-1}) has $0$ or $2^n$ solutions and it has $2^n$ solutions iff $c=0$.
Hence, ${\rm FBET}_{F}(0,0,b,d)={\rm DDT}_{F}(b,d)$.

{\textbf{Case 3:}} Assume that $a\ne 0$, $b=0$. The third equation of (\ref{FBET-n=2m-1-1}) has $0$ or $2^n$ solutions and it has $2^n$ solutions iff $d=0$. Hence,  we have ${\rm FBET}_{F}(a,c,0,0)={\rm DDT}_{F}(a,c)$ if $a\ne 0$.

{\textbf{Case 4:}} Assume that $a\ne 0$, $b=a$. Then (\ref{FBET-n=2m-1-1}) can be reduced to
\begin{equation*}
\arraycolsep=1.2pt\def\arraystretch{1.4}
\left \{\begin{array}{ll}
F(x)+F(x+a)=c,\\
F(x)+F(x+a)=d.
\end{array}\right.
\end{equation*}
Hence, when $a\ne 0$, $b=a$ and $c=d$, we have ${\rm FBET}_{F}(a,c,b,d)={\rm DDT}_{F}(a,c)$, otherwise ${\rm FBET}_{F}(a,c,b,d)=0$.

{\textbf{Case 5:}} Assume that $a\ne 0$, $b\notin \{0, a\}$. According to Theorem \ref{the-FBCT-1}, the first equation of (\ref{FBET-n=2m-1-1}) has no solution. Hence, (\ref{FBET-n=2m-1-1}) has no solution when $a\ne 0$, $b\notin \{0, a\}$. This completes the proof of the case of $n=2m+1$.

For the case of $n=2m$, the proof can be similarly obtained except for the discussion in Case 5. Observe  that (\ref{FBET-n=2m-1-1}) can be rewritten as
\begin{equation}\label{FBET-n=2m-1-2}
\arraycolsep=1.2pt\def\arraystretch{1.4}
\left \{\begin{array}{ll}
F(x)+F(x+a)+F(x+b)+F(x+a+b)=0,\\
F(x)+F(x+a)=c,\\
F(x)+F(x+b)=d.
\end{array}\right.
\end{equation}

By Theorem \ref{the-FBDT-1}, when $a\ne 0$, $b\notin \{0, a\}$, $t_1=1$, and $\frac{a}{b}\in \gf_{2^m}\setminus \{0,1\}$, the first and the second equations of (\ref{FBET-n=2m-1-2}) has $2^m$ solutions, including $x=0, a, b, a+b$. From Theorem \ref{the-DDT}, when $b\ne 0$, $d\ne 0$ and $t_2=1$, we obtain the third equation of (\ref{FBET-n=2m-1-2}) has $2^m$ solutions which include $x=0, b$. Then we obtain that  (\ref{FBET-n=2m-1-2}) has solutions when $a\ne 0$, $b\notin \{0, a\}$, $\frac{a}{b}\in \gf_{2^m}\setminus \{0,1\}$, $t_1=1$ and $t_2=1$. Now we assume that $x\ne 0, a, b, a+b$.
When $a\ne 0$, $b\notin \{0, a\}$, $\frac{a}{b}\in \gf_{2^m}\setminus \{0,1\}$, $t_1=1$ and $t_2=1$, multiplying $x(x+a)(x+b)(x+a+b)$, $x(x+a)$ and $x(x+b)$ on both sides of the first, the second and the third equations of (\ref{FBET-n=2m-1-2}), respectively, one then obtains
\begin{equation}\label{FBET-n=2m-1-3}
\arraycolsep=1.2pt\def\arraystretch{1.4}
\left \{\begin{array}{ll}
a^2bx^{2^{m+1}}+ab^2x^{2^{m+1}}+a^{2^{m+1}}bx^2+ab^{2^{m+1}}x^2=0,\\
a^2x^{2^{m+1}}+a^{2^{m+1}}x^2=0,\\
b^2x^{2^{m+1}}+b^{2^{m+1}}x^2=0.
\end{array}\right.
\end{equation}
According to the discussion of Case 3 in Theorem \ref{the-FBDT-1} when $n=2m$, it can be easily seen that the solutions of the first and the second equations of  (\ref{FBET-n=2m-1-3}) must be the solutions of the third equation of  (\ref{FBET-n=2m-1-3}). Hence, (\ref{FBET-n=2m-1-3}) has $2^m$ solutions when $a\ne 0$, $b\notin \{0, a\}$, $\frac{a}{b}\in \gf_{2^m}\setminus \{0,1\}$, $t_1=1$ and $t_2=1$.

On the other hand, by Theorem \ref{the-FBCT-1}, if $a\ne 0$, $b\notin \{0, a\}$ and $\frac{a}{b}\in \gf_{2^n}\setminus \gf_{2^m}$. then the first equation of  (\ref{FBET-n=2m-1-2}) has no solution. Hence, (\ref{FBET-n=2m-1-2}) has no solution when $a\ne 0$, $b\notin \{0, a\}$ and $\frac{a}{b}\in \gf_{2^n}\setminus \gf_{2^m}$.
According to the DDT of the power mapping $F(x)=x^{2^{m+1}-1}$ in Theorem \ref{the-DDT}, the desired result can be derived.
This completes the proof.
\end{proof}

According to Theorem \ref{the-FBET-1} and Theorem \ref{the-DDT}, we can obtain the following result.

\begin{proposition}\label{cou-FBET-1}
Let $m$ be a positive integer and $F(x)=x^{2^{m+1}-1}$ be a power mapping over $\gf_{2^n}$.
\begin{itemize}
\item [\rm (i)] If $n=2m+1$, then
\begin{equation*}
\begin{aligned}
&|\{(a,c,b,d)\in \gf_{2^n}^4:{\rm FBET}_{F}(a,c,b,d)=2^n\}|=1;\\
&|\{(a,c,b,d)\in \gf_{2^n}^4:{\rm FBET}_{F}(a,c,b,d)=2\}|=3\times 2^{n-1}(2^n-1);\\
&|\{(a,c,b,d)\in \gf_{2^n}^4:{\rm FBET}_{F}(a,c,b,d)=0\}|=2^{4n}-3\times (2^{2n-1}-2^{n-1})-1;
\end{aligned}
\end{equation*}
\item [\rm (i)] If $n=2m$, then
\begin{equation*}
\begin{aligned}
&|\{(a,c,b,d)\in \gf_{2^n}^4:{\rm FBET}_{F}(a,c,b,d)=2^n\}|=1;\\
&|\{(a,c,b,d)\in \gf_{2^n}^4:{\rm FBET}_{F}(a,c,b,d)=2^m\}|=(2^m+1)(2^n-1);\\
&|\{(a,c,b.d)\in \gf_{2^n}^4:{\rm FBET}_{F}(a,c,b,d)=2\}|=3\times (2^{n-1}-2^m)(2^n-1);\\
&|\{(a,c,b,d)\in \gf_{2^n}^4:{\rm FBET}_{F}(a,c,b,d)=0\}|=2^{4n}-(3\cdot 2^{n-1}-2^{m+1}+1)(2^n-1)-1.
\end{aligned}
\end{equation*}
\end{itemize}
\end{proposition}

Consequently, the following result explicitly gives the Feistel boomerang extended uniformity of the power mapping $F(x)=x^{2^{m+1}-1}$ over $\gf_{2^n}$.

\begin{corollary}\label{the-spec-FBET-2}
Let $F(x)=x^{2^{m+1}-1}$ be a power mapping over $\gf_{2^n}$. Then the Feistel boomerang extended uniformity of $F(x)$ satisfies
$$
\beta_e(F)=
\begin{cases}
    2^m, &  {\rm if}\,\,  n=2m;\\
    2, &  {\rm if}\,\, n=2m+1.
\end{cases}
$$
\end{corollary}

\section{Conclusion}\label{con-remarks}

This paper studied differential and Feistel boomerang differential uniformities for vectorial Boolean functions $F$  over binary finite fields.
Specifically,  we presented the DDT, FBCT, FBDT and FBET, respectively, of the power function $F(x)=x^{2^{m+1}-1}$  for positive integer values of  $m$. Our achievements are obtained by solving specific equations over $\gf_{2^n}$ (where $n=2m$ or $n=2m+1$) and by developing techniques to calculate the exact value of each entry on each table and determining the precise number of elements with a given entry.   From the theoretical point of view, our study pushes further former investigations on differential and Feistel boomerang differential uniformities by completing the results presented in the literature by considering new cases of power functions $F$. From a cryptographic point of view, when considering block ciphers (including the Feistel cipher) involving $F$, our in-depth analysis helps decide whether functions $F$  can be regarded as good candidates to be considered or not against differential attacks and Feistel differential and boomerang attacks, respectively.

\section*{Acknowledgments}

This work was supported by the National Key Research and Development Program of China (No. 2021YFA1000600), the National Natural Science Foundation of China (No. 62072162), the Natural Science Foundation of Hubei Province of China (No. 2021CFA079) and the Knowledge Innovation Program of Wuhan-Basic Research (No. 2022010801010319).

\end{document}